\documentclass[twocolumn]{autart}    

\usepackage{graphicx}          
\usepackage{amsmath}
\usepackage{amssymb}
\usepackage{natbib}
\usepackage{comment}
\usepackage{xurl}
\allowdisplaybreaks

\begin{document}

\begin{frontmatter}

\title{Scalable Gaussian Process Regression \\via Deterministic Trigonometric Features: \\Uniform Bounds for Safe Model Predictive Control}%



\thanks[footnoteinfo]{This paper was not presented at any IFAC
meeting. Corresponding author: Julius Jagdt.
Code is available at \\
\url{https://github.com/jagdt/safe-exploration-dtf-GP}.
}

\author[Aachen]{Julius Jagdt}\ead{julius.jagdt@rwth-aachen.de},
\author[Aachen]{Johanna Menn}\ead{johanna.menn@dsme.rwth-aachen.de},
\author[Aachen]{Sebastian Trimpe}\ead{trimpe@dsme.rwth-aachen.de},
\author[ETH]{Melanie N.~Zeilinger}\ead{mzeilinger@ethz.ch},               
\author[Chalmers]{Anna Scampicchio}\ead{anna.scampicchio@chalmers.se}

\address[Aachen]{Institute for Data Science in Mechanical Engineering, RWTH Aachen University, Germany}
\address[ETH]{Institute for Dynamic Systems and Control, ETH Z\"urich, Z\"urich, Switzerland}             
\address[Chalmers]{Department of Electrical Engineering, Chalmers University of Technology, G\"oteborg, Sweden}

\begin{keyword}                           
Model predictive control; Gaussian processes; Safe learning-based control; Scalable Gaussian process regression; Uniform uncertainty bounds.               
\end{keyword}                             

\begin{abstract}                          
Learning-based Model Predictive Control (MPC) using Gaussian processes (GPs) is an effective approach for safe control in the presence of model mismatch. High-probability safety guarantees typically require uncertainty bounds that hold uniformly over the entire state--input domain, but existing bounds are available only for full GP regression. Since exact GP inference scales poorly with the number of data points, its deployment is impractical in large-data regimes. We close this gap by developing a scalable GP framework that admits the derivation of uniform uncertainty bounds. We formalize a deterministic trigonometric feature Gaussian process (DTF-GP), a finite-dimensional kernel approximation based on discretized trigonometric features that reduces GP regression to Bayesian linear regression in feature space. We derive a high-probability uniform uncertainty bound for the proposed DTF-GP and provide its closed-form solution for the squared-exponential kernel case. Finally, we integrate the DTF-GP into a learning-based MPC scheme and demonstrate that it provides high-probability safety guarantees and  exploration performance comparable to a full GP while improving computational efficiency in large-data regimes.
\end{abstract}

\end{frontmatter}

\section{Introduction}


Ensuring safe operation of dynamical systems 
is a central requirement in many control applications. Model Predictive Control (MPC) has become a standard approach for such problems due to its ability to explicitly handle multivariable systems and constraints~\cite{rawlingsModelPredictiveControl2017,borrelliPredictiveControl2017,schwenzerReviewModelPredictive2021}.
Central to the performance and safety of MPC is the accuracy of the predictive model, as model errors can lead to constraint violations and degraded performance. However, in real-world scenarios, first-principles models often suffer from inaccuracies due to unmodeled dynamics and model mismatch.
Learning-based MPC~\cite{hewingLearningbasedModelPredictive2020} addresses this limitation by augmenting nominal models with data-driven ones.


In this context, Gaussian processes (GPs)~\cite{rasmussenGaussianProcessesMachine2006} have emerged as a powerful tool for modeling unknown dynamics due to their nonparametric nature and their inherent uncertainty quantification~\cite{scampicchioGaussianProcessesDynamics2025}.
A key requirement for safety guarantees in learning-based stochastic MPC is the availability of uncertainty bounds that hold uniformly over the entire state–input domain. Uniform bounds for full GP regression have been derived in the literature (see, e.g.~\cite{srinivasGaussianProcessOptimization2010, abbasiyadkoriOnlineLearningLinearly2012,chowdhuryKernelizedMultiarmedBandits2017, fiedlerPracticalRigorousUncertainty2021,whitehouseSublinearRegretGPUCB2023}). These results have enabled GP-based MPC schemes that provide high-probability safety guarantees~\cite{kollerLearningbasedModelPredictive2018,prajapatSafeTractableGaussian2024,prajapatFiniteSampleBasedReachabilitySafe2025}.

However, a fundamental limitation of GP regression remains its poor scalability, as exact inference scales cubically with the number of data points. To alleviate this, a large body of work has proposed scalable GP approximations~\cite{liuWhenGaussianProcess2020}. Existing scalable GP approaches can be categorized as follows:
(i) inducing-variable methods,
(ii) expert-based methods,
(iii) subset-of-data methods, and
(iv) methods based on finite-dimensional kernel representations.
Inducing-variable methods (i) introduce a set of auxiliary variables to construct low-rank approximations of the kernel matrix~\cite{titsiasVariationalLearningInducing2009,hensmanGaussianProcessesBig2013,lazaro-gredillaInterdomainGaussianProcesses2009}. For this class, convergence guarantees of the approximate posterior and pointwise error bounds have been established~\cite{burtConvergenceSparseVariational2020}, yet no uniform uncertainty bounds are available.
Expert-based methods (ii)~\cite{trespMixturesGaussianProcesses2000,yukselTwentyYearsMixture2012,trappDeepStructuredMixtures2020} and subset-of-data methods (iii)~\cite{hayashiRandomSubsamplingGaussian2020} partition either the input space or the data set. The resulting local models make it difficult to derive uncertainty bounds that hold uniformly over the entire input domain.

Methods based on finite-dimensional kernel representations (iv) approximate the kernel using a finite set of features, thereby reducing GP regression to Bayesian linear regression in a finite-dimensional feature space. 
For stationary kernels, trigonometric features arise naturally from the spectral representation given by Bochner's theorem~\cite{steinInterpolationSpatialData1999}. A key design aspect in these approaches is the selection of the feature frequencies and their associated coefficients.
A prominent example is the Random Fourier Features method~\cite{rahimiRandomFeaturesLargescale2007}, which samples the frequencies from the kernel’s spectral density. For Random Fourier Features, probabilistic guarantees have been established for the expected risk~\cite{rudiGeneralizationPropertiesLearning2017} and for the approximation quality of the kernel~\cite{sriperumbudurOptimalRatesRandom2015}. Quadrature-based approaches~\cite{daoGaussianQuadratureKernel2017,mutnyEfficientHighDimensional2018,liTrigonometricQuadratureFourier2024} use numerical integration rules to determine the frequencies and coefficients, yielding deterministic error bounds on the kernel approximation. However, as with Random Fourier Features, no guarantees are available for the posterior function approximation itself. Sparse Spectrum GP regression~\cite{lazaro-gredillaSparseSpectrumGaussian2010} treats the frequencies as hyperparameters and optimizes them. Although tuning the frequencies achieves good function approximation in practice, the optimization yields only local optima, which prevents establishing theoretical guarantees on the posterior accuracy and uncertainty quantification.
In summary, to the best of our knowledge, none of the existing scalable GP approaches provide uniform uncertainty bounds for the posterior function estimate. As a consequence, they cannot be integrated into safety-critical MPC frameworks that rely on such bounds.



This paper addresses this gap by developing the deterministic trigonometric feature Gaussian process (DTF-GP), a scalable GP framework that retains the ability to provide high-probability uniform uncertainty bounds. The main contributions of this paper are as follows:
(i) we formalize the DTF-GP, a finite-dimensional kernel approximation using trigonometric features with frequencies selected deterministically from a uniformly spaced frequency grid;
(ii) we derive a high-probability uniform uncertainty bound for the DTF-GP; 
(iii) we demonstrate how the proposed model can be integrated into a stochastic MPC framework to enable safe exploration with reduced computational complexity in large-data regimes.

\section{Problem statement}
\label{sec:problem_statement}

We consider a nonlinear, discrete-time dynamical system described by the nonlinear difference equation
\begin{equation*}
\label{eq:system_dynamics}
    x_{i+1} = f(x_i,u_i) = \underbrace{h(x_i,u_i)}_{\text{nominal model}} + \underbrace{g(x_i, u_i)}_{\text{model error}},
\end{equation*}
where $x_i \in \mathbb{R}^{n_x}$ denotes the system state and \mbox{$u_i \in \mathbb{R}^{n_u}$} the control input at time $i$. 
The function $h$ represents a known nominal model (e.g., a first-principles physical model) assumed to be twice continuously differentiable. 
The model error $g$ captures the a priori unknown model mismatch and is to be learned from data.
The system is subject to state and input constraints $x_i \in \mathcal X \subset \mathbb{R}^{n_x}$ and $u_i \in \mathcal{U} \subset \mathbb{R}^{n_u}$, which encode safety or operational requirements 
(e.g., state constraints could represent admissible regions, while input constraints could capture actuator limits). 
We define the joint state--input domain as $\mathcal{Z} \doteq \mathcal{X} \times \mathcal{U} \subset \mathbb{R}^d$, where $d = n_x + n_u$, and assume that $\mathcal{Z}$ is compact.
For $z_i=(x_i^\top,u_i^\top)^\top \in \mathcal Z$ and given noisy observations of state transitions $\hat{x}_{i+1}=f(z_i)+\varepsilon_i$ we learn the model error $g$ from residual observations
\begin{equation}
\label{eq:model_error_observations}
y_i 
= 
\hat{x}_{i+1} - h(z_i)
=
g(z_i) + \varepsilon_i, \qquad i=1,\dots,N,
\end{equation}
where the noise sequence $(\varepsilon_i)_{i\ge1}$ is assumed to be conditionally $R$-sub-Gaussian.

We model $g$ by a GP prior with kernel $k$.
A common choice is the squared-exponential kernel
\begin{equation}
\label{eq:rbf_kernel}
    k_{\mathrm{RBF}}(z,z')
    =
    \sigma_f^2
    \exp\!\left(
    -\frac{1}{2}(z-z')^\top \Lambda^{-1}(z-z')
    \right),
\end{equation}
where $\sigma_f^2>0$ denotes the prior variance and
$\Lambda=\operatorname{diag}(\ell_1^2,\dots,\ell_d^2)$ contains the squared lengthscales. 
This kernel serves as the primary example throughout the paper.

The proposed scalable GP model is based on a finite-dimensional feature representation
\[
    k(z,z') \approx \tilde k(z,z') = \phi(z)^\top\phi(z'),
    \qquad \phi(z)\in\mathbb R^M.
\]
Such an approximation allows us to reduce GP regression into Bayesian linear regression in feature space. Introducing the feature matrix and the vector of output measurements
\[
    \Phi_N \doteq
    \begin{bmatrix}
    \phi(z_1)^\top \\
    \vdots \\
    \phi(z_N)^\top
    \end{bmatrix},
    \qquad
    y_{1:N}\doteq(y_1,\dots,y_N)^\top,
\]
and defining, for observation noise variance $\sigma_n^2$, 
\begin{equation}
\label{eq:regularized_feature_matrix}
    V_N \doteq \Phi_N^\top\Phi_N+\sigma_n^2 I,
\end{equation}
the scalable GP model features a Gaussian posterior distribution with mean and variance given by 
\begin{align}
\label{eq:feature_posterior_mean}
    \mu_N(z)
    &=
    \phi(z)^\top V_N^{-1}\Phi_N^\top y_{1:N},\\
\label{eq:feature_posterior_variance}
    \sigma_N^2(z)
    &=
    \sigma_n^2\,\phi(z)^\top V_N^{-1}\phi(z).
\end{align}
In order to provide safety guarantees in learning-based MPC, a high-probability uniform uncertainty bound of the form
\begin{equation*}
\left|g(z) - \mu_N(z) \right| \le \nu_N(z), \quad \forall z \in \mathcal{Z},
\end{equation*} 
is required, where $\nu_N(\cdot)$ is a computable confidence bound.

\section{Scalable Gaussian process regression via deterministic trigonometric features}
\label{sec:scalable_gp}

We now introduce the proposed scalable GP model, referred to as a deterministic trigonometric feature Gaussian process (DTF-GP).
We begin by reviewing the spectral representation of stationary kernels in Section~\ref{sec:spectral_representation}, which motivates the proposed approach.
Based on this representation, we present in Section~\ref{sec:scalable_kernel} a finite-dimensional trigonometric kernel approximation using a fixed deterministic frequency grid.
In contrast to randomized or optimized frequency selection approaches such as Random Fourier Features~\cite{rahimiRandomFeaturesLargescale2007} and Sparse Spectrum GP regression~\cite{lazaro-gredillaSparseSpectrumGaussian2010}, this deterministic construction enables the derivation of the uniform uncertainty bound (presented in Section~\ref{sec:main_bound}) required for safety guarantees in learning-based MPC schemes.
We then introduce a structured spectral parametrization and a corresponding frequency selection strategy for approximating the squared-exponential (RBF) kernel in Sections~\ref{sec:spectral_parametrization} and~\ref{sec:frequency_selection}. While these choices are tailored to the RBF kernel, the general ideas of spectral parametrization and frequency truncation can in principle be extended to other stationary kernels. Finally, in Section~\ref{sec:hypothesis_space}, we characterize the associated hypothesis space through an RKHS interpretation of the constructed kernel, which forms the basis for the uncertainty analysis in the subsequent section.

\subsection{Spectral representation of stationary kernels}
\label{sec:spectral_representation}

By Bochner’s theorem~\cite{steinInterpolationSpatialData1999}, any continuous stationary kernel admits the spectral representation
\begin{equation*}
\label{eq:bochner_density}
k(\tau) \;=\; \int_{\mathbb{R}^d} e^{i2\pi \omega^\top \tau}\, S(\omega)\,\mathrm{d}\omega,
\end{equation*}
where $\tau=z-z'$ and $S(\omega)$ denotes the spectral density of the kernel.
Since $k$ is real-valued, the kernel can be expressed in terms of trigonometric functions as
\begin{equation}
\label{eq:trig_motivation}
\begin{aligned}
k(z-z')
&= \int_{\mathbb{R}^d} \Bigl( \cos(2\pi \omega^\top z)\cos(2\pi \omega^\top z') \\
&\quad + \sin(2\pi \omega^\top z)\sin(2\pi \omega^\top z')\Bigr)\, S(\omega)\,\mathrm{d}\omega .
\end{aligned}
\end{equation}
This representation shows that stationary kernels can be viewed as continuous superpositions of trigonometric basis functions weighted by the spectral density, and naturally motivates using weighted trigonometric features for the finite-dimensional approximation of the kernel.

\subsection{Finite-dimensional kernel approximation}
\label{sec:scalable_kernel}

We approximate the spectral integral~\eqref{eq:trig_motivation} by a finite weighted sum of trigonometric basis functions at selected frequencies.
Due to symmetry of the trigonometric features, one representative from each symmetric frequency pair is sufficient.
We therefore define the half lattice
\begin{equation}
\label{eq:half_lattice}
\mathbb Z^d_+
\doteq
\left\{
q\in\mathbb Z^d\setminus\{0\}
:\,
q_j>0
\text{ for } j=\min\{i:q_i\neq0\}
\right\}.
\end{equation}
Thus, for every nonzero pair $\{q,-q\}$, exactly one representative is contained in $\mathbb Z^d_+$.
For a period vector $T=(T_1,\dots,T_d)\in\mathbb R^d_{>0}$, we define a frequency grid as
\begin{equation}
\label{eq:omega_q}
\omega_q = \left(\frac{q_1}{T_1},\dots,\frac{q_d}{T_d}\right), 
\qquad q\in\mathbb Z^d_+ .
\end{equation}
Selecting a finite subset of frequencies $\{\omega^{(1)},\dots,\omega^{(Q)}\}$ from the grid, we define the feature map
\begin{equation}
\label{eq:phi_trig_features}
\phi(z) \doteq
\begin{bmatrix}
\sqrt{\lambda^{(0)}} \\
\sqrt{\lambda^{(1)}}\sqrt{2}\cos(2\pi (\omega^{(1)})^\top z) \\
\sqrt{\lambda^{(1)}}\sqrt{2}\sin(2\pi (\omega^{(1)})^\top z) \\
\vdots \\
\sqrt{\lambda^{(Q)}}\sqrt{2}\cos(2\pi (\omega^{(Q)})^\top z) \\
\sqrt{\lambda^{(Q)}}\sqrt{2}\sin(2\pi (\omega^{(Q)})^\top z)
\end{bmatrix}
\in \mathbb R^{M},
\end{equation}
where \(\lambda^{(s)} > 0\) are spectral weights associated with the selected frequencies $\omega^{(s)}$ and $M=2Q+1$ denotes the number of features. This induces the finite-dimensional kernel
\begin{equation*}
\tilde{k}_\mathrm{DTF}(z,z') = \phi(z)^\top \phi(z').
\end{equation*}
The resulting kernel can be interpreted as a discretization of the spectral integral~\eqref{eq:trig_motivation}, where the continuous spectral density is approximated by weighted point masses at the selected frequencies.
This interpretation becomes transparent when the kernel is written component-wise:
\begin{equation}
\label{eq:finite_kernel_sum}
\tilde{k}_\mathrm{DTF}(z,z') 
= \sum_{m=0}^{M-1} \lambda^{(\lceil m/2\rceil)}\,
\varphi_m(z)\,\varphi_m(z'),
\end{equation}
with
\begin{equation}
\label{eq:trig_basis_functions}
\varphi_m(z) =
\begin{cases}
\qquad\qquad 1,
& \text{if } m=0, \\[6pt]
\sqrt{2}\,\cos\!\bigl( 2\pi\,(\omega^{(\lceil m/2\rceil)})^\top z \bigr),
& \text{if } m \text{ is odd}, \\[6pt]
\sqrt{2}\,\sin\!\bigl( 2\pi\,(\omega^{(\lceil m/2\rceil)})^\top z \bigr),
& \text{if } m \text{ is even}.
\end{cases}
\end{equation}
In particular, the weights \(\lambda^{(s)}\) correspond to a discretized spectral density and determine the smoothness properties of the resulting function class.

\subsection{Spectral parameterization}
\label{sec:spectral_parametrization}

In principle, one could treat each weight $\lambda^{(s)}$ as an independent hyperparameter. However, this would lead to a large number of hyperparameters.
To alleviate this, we impose the decay structure of the spectral density $S(\omega)$ of the target kernel on the weights. Specifically, we define weights $\lambda_q$ on the full frequency grid, where $\lambda_q$ denotes the weight associated with the frequency $\omega_q$. For the squared-exponential (RBF) kernel, the spectral density is Gaussian and given by
\[
S_{\mathrm{RBF}}(\omega)
=
\sigma_f^2 (2\pi)^{d/2}\det(\Lambda)^{1/2}
\exp\!\left(
-2\pi^2 \omega^\top \Lambda \omega
\right).
\]
Motivated by this exponential spectral decay, we choose
\begin{equation}
\label{eq:lambda_decay}
\lambda_q = C \exp(-\omega_q^\top D \omega_q),
\end{equation}
with hyperparameters $C > 0$ and diagonal matrix $D=\operatorname{diag}(a_1,\dots,a_d) \succ 0$. This parametrization mirrors the role of the prior variance and lengthscales in the original RBF kernel, preserving both the number and interpretation of hyperparameters. In particular, the scaling factor $C$ controls the overall spectral energy and corresponds to the prior variance, while the matrix $D$ determines the rate of spectral decay and thus governs the smoothness of functions in the induced RKHS, analogous to the role of lengthscales in the RBF kernel.
Moreover, the structured decay enables the derivation of an explicit error bound (cf.~Section~\ref{sec:projection_error_rbf}).

\subsection{Frequency selection}
\label{sec:frequency_selection}

The finite representation is obtained by truncating the frequency grid to the most relevant components. Since the weights defined in~\eqref{eq:lambda_decay} decay with $\omega^\top D \omega$, where the frequencies are given by~\eqref{eq:omega_q}, their level sets are ellipsoids in frequency space. We therefore retain frequencies within an ellipsoidal region
\begin{equation}
\label{eq:frequency_selection}
q^\top \widetilde D q \le r^2,
\qquad
\widetilde D \doteq \mathrm{diag}\!\left(\frac{a_1}{T_1^2},\dots,\frac{a_d}{T_d^2}\right),
\end{equation}
for a suitable truncation radius $r$. For a given number of retained grid frequencies, this selection keeps the frequencies with the largest spectral weights and hence minimizes the discarded spectral mass.


\subsection{Hypothesis space}
\label{sec:hypothesis_space}

The DTF-GP corresponds to truncating an infinite-dimensional reproducing kernel Hilbert space (RKHS) $\mathcal H_{k_\mathrm{DTF}}$ induced by the DTF kernel $k_\mathrm{DTF}$, which admits the representation
\begin{equation}
\label{eq:infinite_kernel}
k_\mathrm{DTF}(z,z') = \sum_{m=0}^\infty \lambda^{(\lceil m/2\rceil)}\,
\varphi_m(z)\varphi_m(z').
\end{equation}
Any function $g \in \mathcal H_{k_\mathrm{DTF}}$ admits an expansion with respect to the orthonormal system $\{\varphi_m\}_{m\ge1}$
\begin{equation}
\label{eq:function_expansion}
g(z) = \sum_{m=0}^\infty \alpha_m \varphi_m(z),
\end{equation}
with coefficients $(\alpha_m)_{m\ge1}$ satisfying
\begin{equation}
\label{eq:rkhs_norm}
\|g\|_{\mathcal H_{k_\mathrm{DTF}}}^2
= \sum_{m=0}^\infty \frac{\alpha_m^2}{\lambda^{(\lceil m/2\rceil)}} < \infty.
\end{equation}
Note that, on the compact domain $\mathcal{Z}=\prod_{j=1}^d [0,T_j]$ equipped with the normalized Lebesgue measure $d\mu(z)=\frac{1}{|\mathcal{Z}|}dz$, the basis functions $\{\varphi_m\}_{m\ge0}$ form an orthonormal system and the weights $\lambda^{(\lceil m/2\rceil)}$ are the Mercer eigenvalues of the DTF kernel. Further details are provided in Appendix~\ref{app:mercer}.

The finite kernel~\eqref{eq:finite_kernel_sum} induces a finite-dimensional RKHS $\mathcal{H}_{\tilde{k}_\mathrm{DTF}} \subset \mathcal H_{k_\mathrm{DTF}}$, spanned by the truncated $M$ basis functions.
To relate these spaces, we introduce the orthogonal projection
\begin{equation}
\label{eq:projection_operator}
    P : \mathcal H_{k_\mathrm{DTF}} \to \mathcal H_{\tilde{k}_\mathrm{DTF}}.
\end{equation}
In terms of the expansion~\eqref{eq:function_expansion}, the projection is given by
\begin{equation}
\label{eq:projection}
(Pg)(z) = \sum_{m=0}^{M-1} \alpha_m \varphi_m(z).
\end{equation}

In the following, we will assume that the unknown function $g$ belongs to the RKHS $\mathcal H_{k_\mathrm{DTF}}$ induced by the infinite-dimensional kernel. This infinite-dimensional RKHS therefore serves as the hypothesis space for the derivation of our uniform uncertainty bound.

\begin{rem}[Periodicity]
\label{rem:periodicity}
Due to the use of a discrete frequency grid, the resulting kernel $k_\mathrm{DTF}$ is $T=(T_1,\dots,T_d)$-periodic in each coordinate. For a suitable choice of the weights, $k_\mathrm{DTF}$ coincides with the Fourier series representation of a periodized stationary kernel. A derivation of this equivalence is provided in Appendix~\ref{app:periodization}.
While this introduces an implicit periodicity assumption, its effect on non-periodic functions is negligible in practice: as analyzed in Appendix~\ref{app:periodic_extension}, if the period $T$ is chosen sufficiently larger than the domain of interest, the induced approximation error decays rapidly (exponentially for RBF kernels) with the margin between the domain and the period.
\end{rem}

\section{Uniform uncertainty bound}\label{sec:main_bound}

In this section, we derive a high-probability uniform uncertainty bound for the DTF-GP. We do so by making use of the projection operator introduced in~\eqref{eq:projection_operator}, which relates the infinite-dimensional hypothesis space to its finite-dimensional approximation. 
We first present the main uniform uncertainty bound in Section~\ref{sec:uniform_uncertainty_bound}, which applies to general finite-dimensional kernel approximations. We then derive an explicit projection error bound for the proposed approximation of the RBF kernel in Section~\ref{sec:projection_error_rbf}. Finally, Section~\ref{sec:numerical_illustration} illustrates the behavior of the resulting confidence bound, in particular the effect of the number of retained frequencies on its conservatism.

The uncertainty bound derived in this section is of \emph{a posteriori} type, i.e., it depends on quantities that are computed from the observed data after posterior inference. In contrast to GP confidence bounds based on information gain (e.g.,~\cite{srinivasGaussianProcessOptimization2010,chowdhuryKernelizedMultiarmedBandits2017}), which are \emph{a priori} and independent of the realized data, a posteriori bounds exploit the structure of the observed data and are typically less conservative. This makes them well suited for control applications, where posterior quantities are naturally available after inference and the confidence bounds can be refined online as new data are collected~\cite{fiedlerPracticalRigorousUncertainty2021}.

\subsection{Uniform uncertainty bound}
\label{sec:uniform_uncertainty_bound}

In the following, we present our main result.

\begin{thm}
\label{thm:uniform_bound}

Let $\mathcal Z \subset \mathbb R^d$ be compact and $(\mathcal F_i)_{i\ge 0}$ be a filtration. 
Let $(z_i)_{i\ge 1}$ be a $\mathcal Z$-valued stochastic process that is predictable with respect to $(\mathcal F_i)_{i\ge 0}$. 
Consider the measurements model $y_i = g(z_i) + \varepsilon_i, \,\forall\, i\ge 1$, where $(\varepsilon_i)_{i\ge 1}$ is adapted to $(\mathcal F_i)_{i\ge 0}$ and conditionally $R$-sub-Gaussian given $\mathcal F_{i-1}$.
Let $k_\mathrm{DTF}$ be a positive definite kernel with corresponding RKHS $\mathcal H_{k_\mathrm{DTF}}$ as introduced in Section~\ref{sec:scalable_gp}.
Let $Pg$ denote the projection defined in~\eqref{eq:projection}, and let $\tilde{k}_\mathrm{DTF}$ be the associated finite-dimensional kernel. 
Assume that $g \in \mathcal H_{k_\mathrm{DTF}}$ with $\|g\|_{\mathcal H_{k_\mathrm{DTF}}} \le B$, and let $\mu_N(\cdot)$ and $\sigma_N(\cdot)$ denote the posterior mean and standard deviation of the GP with kernel $\tilde{k}_\mathrm{DTF}$ as defined in~\eqref{eq:feature_posterior_mean} and~\eqref{eq:feature_posterior_variance}, respectively, with $V_N$ as in~\eqref{eq:regularized_feature_matrix}.
Then, for any $\delta \in (0,1]$, with probability at least $1-\delta$,
uniformly for all $N \ge 1$ and all $z \in \mathcal Z$,
\begin{equation}
\label{eq:final_uniform_bound}
\bigl|g(z)-\mu_N(z)\bigr|
\;\le\;
\bigl| g(z) - Pg(z) \bigr|
+
\sigma_N(z)\,\beta_N,
\end{equation}
where
\begin{equation*}
\label{eq:beta_N}
\begin{aligned}
\beta_N =
&\;B
\;+\;\frac{R}{\sigma_n}
\sqrt{2\ln\!\left(
\frac{1}{\delta}
\sqrt{\det\!\left(I+\frac{1}{\sigma_n^2}\Phi_N^\top\Phi_N\right)}
\right)} \\
&+\; \frac{1}{\sigma_n}
\bigl\|\Phi_N^\top\!\bigl(Pg_{1:N}-g_{1:N}\bigr)\bigr\|_{V_N^{-1}}.
\end{aligned}
\end{equation*}

\end{thm}

\begin{pf}
The proof follows the approach of
~\cite{chowdhuryKernelizedMultiarmedBandits2017} by decomposing the error into a noise term and an approximation term.
The noise term is bounded using a self-normalized concentration inequality from
~\cite{abbasiyadkoriOnlineLearningLinearly2012}. 
The approximation term is handled by introducing the projection $P$ onto the finite-dimensional RKHS. 
The detailed proof is given in Appendix~\ref{app:proof_uniform_bound}. \hfill$\blacksquare$
\end{pf}

For comparison, the corresponding high-probability confidence bound for a full GP takes the form
\begin{equation}
\label{eq:full_gp_bound}
\bigl|g(z)-\mu_N^\mathrm{full}(z)\bigr| 
\le \beta_N^\mathrm{full} \, \sigma_N^\mathrm{full}(z),
\end{equation}
with
\begin{equation*}
\label{eq:full_gp_beta}
\beta_N^\mathrm{full} = \;B
\;+\; \frac{R}{\sigma_n}\,\sqrt{ 2 \ln \!\left(
    \frac{1} { \delta }  \sqrt{\det\!\left(I + \frac{1}{\sigma_n^2} K_N\right)}
    \right) },
\end{equation*}
where $K_N$ denotes the Gram matrix.
This bound follows from combining a self-normalized concentration inequality 
(see~\cite[Corollary 3.5]{abbasiyadkoriOnlineLearningLinearly2012}) and~\cite[Corollary 1]{whitehouseSublinearRegretGPUCB2023}) with an RKHS norm bound on the approximation term~\cite{chowdhuryKernelizedMultiarmedBandits2017}. 

The bound in Theorem~\ref{thm:uniform_bound} has the same overall structure as the full GP confidence bound~\eqref{eq:full_gp_bound}. In particular, the noise term in Theorem~\ref{thm:uniform_bound} is the feature-space analogue of the noise term in the full GP bound. However, the proposed bound contains two additional terms that account for the finite-dimensional kernel approximation: the projection error $\bigl|g(z)-Pg(z)\bigr|$ and an additional term in $\beta_N$ that depends on the projection residual $Pg_{1:N}-g_{1:N}$. The latter can be upper bounded in terms of the maximum projection error $\sup_{z\in\mathcal Z}\bigl|g(z)-Pg(z)\bigr|$. Consequently, the proposed DTF-GP confidence bound depends on the same quantities as the full GP confidence bound~\eqref{eq:full_gp_bound}, namely the RKHS norm, the sub-Gaussian parameter $R$, and posterior quantities, together with an upper bound on the projection error.

\begin{rem}[RKHS norm bound]
As in full GP confidence bounds, the proposed bound requires an upper bound $B$ on the RKHS norm of the unknown model error. The problem of obtaining such a bound is well known in the literature; see, e.g.,~\cite{fiedlerSafetySafeBayesian2024} for a discussion. Since this issue is independent of the DTF-GP approximation, it is not addressed in this paper. We refer to, e.g.,~\cite{tokmakAutomaticNonlinearMPC2025,tokmakPACSBOProbablyApproximately2024,tokmakSafeExplorationReproducing2025,wenzelReliableSamplingbasedRKHS2026} for recent work on estimating an upper bound on the RKHS norm.
\end{rem}

\subsection{Projection error bound for RBF kernel}
\label{sec:projection_error_rbf}

To evaluate the uniform uncertainty bound in Theorem~\ref{thm:uniform_bound}, an upper bound on the maximum projection error $\sup_{z\in\mathcal Z}\bigl|g(z)-Pg(z)\bigr|$ is required. We present an explicit upper bound for this term for the RBF spectral parameterization~\eqref{eq:lambda_decay} and the ellipsoidal frequency selection~\eqref{eq:frequency_selection}, which were introduced to approximate the RBF kernel.

\begin{prop}
\label{prop:projection_error_rbf}
Assume that $g \in \mathcal H_{k_\mathrm{DTF}}$ with $\|g\|_{\mathcal H_{k_\mathrm{DTF}}} \le B$. Let the spectral weights be
given by
\[
\lambda_q = C \exp(-q^\top \widetilde D q),
\]
with $C>0$ and $\widetilde D \succ 0$, and let the retained frequency set be
\[
\Omega_r = \{\omega_q \mid q^\top \widetilde D q \le r^2\}.
\]
Define
\[
\rho \doteq \frac{1}{2}\sqrt{\operatorname{tr}(\widetilde D)}
\]
and let
\[
S_{d-1} \doteq \frac{2\pi^{d/2}}{\Gamma(d/2)}
\]
denote the surface area of the unit sphere in $\mathbb R^d$. Then
\begin{align*}
&\sup_{z\in\mathcal Z}|g(z)-Pg(z)| \\
& \qquad \quad \le
2B
\Bigg(
\frac{C}{\sqrt{\det(\widetilde D)}}\,
S_{d-1}
\int_{r-\rho}^{\infty}
e^{-(t-\rho)^2}t^{d-1}\,\mathrm dt
\Bigg)^{1/2}.
\end{align*}
\end{prop}

\begin{pf}
In a first step, the projection error is bounded in terms of the RKHS norm and the sum of discarded weights
\[
\sup_{z\in\mathcal Z}|g(z)-Pg(z)|
\le
2B
\left(
\sum_{q^\top \widetilde D q > r^2} \lambda_q
\right)^{1/2}.
\]
Using the parametrization of the weights~\eqref{eq:lambda_decay}, the remaining sum is turned into an integral via a lattice argument. After a change of variables, this yields an integral over $\mathbb R^d$, which is reduced to a one-dimensional integral using polar coordinates. The detailed derivation is given in Appendix~\ref{app:proof_proj_error}. \hfill$\blacksquare$
\end{pf}

The bound in Proposition~\ref{prop:projection_error_rbf} depends only on the RKHS norm, the parameters $C$ and $\widetilde D$, and the radius $r$ used to truncate the frequencies. The remaining integral can be evaluated numerically in practice. Consequently, when combined with Theorem~\ref{thm:uniform_bound}, the resulting DTF-GP uncertainty bound requires only the RKHS norm, the sub-Gaussian parameter $R$, and quantities that are known a posteriori. Hence, for the considered approximation of the RBF kernel, the DTF-GP uniform uncertainty bound relies on the same requirements as full GP confidence bounds.

\subsection{Numerical illustration}
\label{sec:numerical_illustration}

We illustrate the behavior of the proposed uniform uncertainty bound and compare it to the corresponding full GP confidence bound~\eqref{eq:full_gp_bound}. We consider a one-dimensional regression problem on the compact domain $\mathcal Z = [-5,5]$, where $200$ noisy observations are sampled uniformly from $\mathcal Z$ with noise variance $\sigma_n^2=0.04$.
The ground-truth function is generated as a finite RBF-kernel expansion
\[
g(z)
=
\sum_{i=1}^{N_s} w_i k_{\mathrm{RBF}}(s_i,z),
\]
where the support points $s_i$ are sampled uniformly from $\mathcal Z$ and the coefficients $w_i$ are drawn from a standard normal distribution.
The RBF kernel used to generate $g_{\mathrm{true}}$ has prior variance $\sigma_f^2=0.5$ and lengthscale $\ell=0.5$.
The DTF-GP is compared to a full GP using hyperparameters learned by maximizing the marginal likelihood. For the DTF-GP, the finite-dimensional approximation is constructed as introduced in Section~\ref{sec:scalable_gp}, with period $T=15$.
All bounds are evaluated with confidence level $1-\delta=0.95$. 
For the DTF-GP confidence bound, the projection error is computed using Proposition~\ref{prop:projection_error_rbf}. Since the ground-truth function is known in this controlled experiment, the RKHS norm bounds used in the confidence bounds are computed from the true underlying function.

To interpret the contributions to the bound width, we decompose the DTF-GP confidence bound as
\begin{equation}
\label{eq:dtf_bound_decomposition}
\nu_N(z)
=
\mathcal B_{\mathrm{RKHS}}(z)
+
\mathcal B_{\mathrm{noise}}(z)
+
\mathcal B_{\mathrm{res}}(z)
+
\mathcal B_{\mathrm{proj}}(z),
\end{equation}
where
\begin{align*}
\mathcal B_{\mathrm{RKHS}}(z)
&:=
\sigma_N(z)B, \\
\mathcal B_{\mathrm{noise}}(z)
&:=
\sigma_N(z)\frac{R}{\sigma_n}
\sqrt{2\ln\!\biggl(
\frac{1}{\delta}
\sqrt{\det\!\bigl(I+\frac{1}{\sigma_n^2}\Phi_N^\top\Phi_N\bigr)}
\biggr)}, \\
\mathcal B_{\mathrm{res}}(z)
&:=
\sigma_N(z)\frac{1}{\sigma_n}
\bigl\|\Phi_N^\top(Pg_{1:N}-g_{1:N})\bigr\|_{V_N^{-1}}, \\
\mathcal B_{\mathrm{proj}}(z)
&:=
|g(z)-Pg(z)|.
\end{align*}
For the full GP bound~\eqref{eq:full_gp_bound}, only the corresponding RKHS and noise contributions are present.

Figure~\ref{fig:bound_width_vs_freq} shows the average width of the resulting uncertainty bounds as a function of the number of retained frequencies.
The bars are decomposed according to the contributions defined above, with $\mathcal B_{\mathrm{res}}(z)$ and $\mathcal B_{\mathrm{proj}}(z)$ appearing only for the DTF-GP.
The results are consistent with the structure of the bound: the additional conservatism mainly depends on the number of retained frequencies (cf. Proposition~\ref{prop:projection_error_rbf}). 
Increasing the number of retained frequencies reduces the spectral truncation error and therefore decreases the projection-dependent terms $\mathcal B_{\mathrm{res}}(z)$ and $\mathcal B_{\mathrm{proj}}(z)$.
In the shown experiment, the term $\mathcal B_{\mathrm{proj}}(z)$ is included in the stacked bars but is small compared to the other contributions and therefore barely visible on the plotted scale.
In the limit of sufficiently many features, both projection-dependent terms vanish, and the DTF-GP confidence bound approaches the full GP confidence bound.

\begin{figure}
    \centering
    \includegraphics[width=0.45\textwidth]{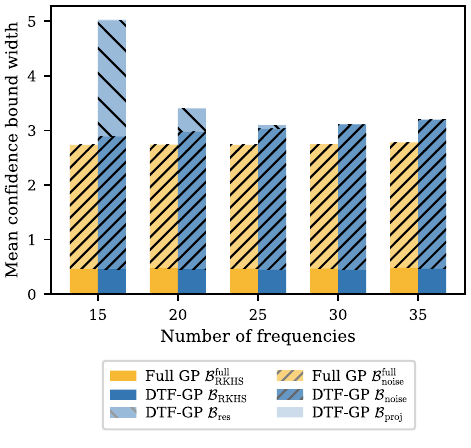}
    \caption{Uncertainty bound width averaged over $100$ random seeds as a function of the number of retained frequencies. The stacked bars show the contributions in~\eqref{eq:dtf_bound_decomposition}; for the full GP, the corresponding RKHS and noise contributions are defined and shown analogously.}
    \label{fig:bound_width_vs_freq}
\end{figure}

\section{Closing the loop: Safe learning-based MPC with DTF-GP}
\label{sec:safe_mpc}

We now integrate the proposed DTF-GP and the derived uniform uncertainty bound into a safe learning-based MPC framework. 
Section~\ref{sec:mpc_integration} describes the integration of the DTF-GP confidence bound into the SafeMPC formulation of
~\cite{kollerLearningbasedModelPredictive2018}.
Section~\ref{sec:safe_exploration} then introduces the resulting safe exploration strategy. 
Implementation details of the GP models, including kernel design, frequency selection, hyperparameter learning, and online updates, are discussed in Section~\ref{sec:mpc_implementation}. Finally, Section~\ref{sec:mpc_results} evaluates the proposed approach on a nonlinear control problem and compares it to the one based on a full GP in terms of exploration performance and computational efficiency.

\subsection{Integration into SafeMPC}
\label{sec:mpc_integration}

We now integrate the proposed DTF-GP confidence bound into the SafeMPC framework of~\cite{kollerLearningbasedModelPredictive2018}.
Following~\cite{kollerLearningbasedModelPredictive2018}, we assume the existence of a backup controller and a safe region.

\begin{assum}
\label{ass:safe_backup_controller}
We are given a backup controller $\pi_{\mathrm{safe}}$ and a safe region $\mathcal X_{\mathrm{safe}} \subseteq \mathcal X$ that is robust control positively invariant under $\pi_{\mathrm{safe}}$.
Moreover, the backup controller satisfies the input constraints inside $\mathcal X_{\mathrm{safe}}$, i.e., $\pi_{\mathrm{safe}}(x)\in\mathcal U \text{ for all } x\in\mathcal X_{\mathrm{safe}}$.
\end{assum}

At time~$i$, the controller propagates uncertainty sets $R_t$ over a finite prediction horizon of length $H$ and solves
\begin{subequations}
\label{eq:safempc_problem}
\begin{align}
\min_{\pi_0,\dots,\pi_{H-1}} \quad & J(R_0,\dots,R_H) \\
\text{s.t.}\quad
& R_{t+1} = \mathcal M(R_t,\pi_t), \quad t=0,\dots,H-1, \\
& R_t \subset \mathcal X, \quad t=1,\dots,H-1, \\
& \pi_t(R_t) \subset \mathcal U, \quad t=0,\dots,H-1, \\
& R_H \subset \mathcal X_{\mathrm{safe}},
\end{align}
\end{subequations}
where $R_0=\{x_i\}$, $J$ denotes a cost functional,  $\pi_t$ the optimized policy at prediction step $t$, and $\mathcal M$ denotes the uncertainty propagation map.

The propagation map $\mathcal M$ is obtained by combining a local linearization of the nominal model with confidence bounds on the learned model error $g$. It over-approximates the possible successor states by accounting for the linearization error, a Lipschitz error term on the model error $g$, and the GP confidence bound; see~\cite{kollerLearningbasedModelPredictive2018} for the detailed construction.
Thus, the role of the GP model is to provide high-probability confidence intervals for the model error $g$. In the original SafeMPC formulation, these bounds are obtained from a full GP. Here, we replace the full GP by the DTF-GP from Section~\ref{sec:scalable_gp}. For each state dimension, an independent GP is trained on observations of the corresponding component of the model error, as introduced in Section~\ref{sec:problem_statement}.
The confidence bound from Theorem~\ref{thm:uniform_bound}, together with the projection error bound from Proposition~\ref{prop:projection_error_rbf}, is then applied component-wise and used in $\mathcal M$ to propagate the confidence sets.

\begin{cor}
\label{cor:dtf_safempc_safety}
Suppose Assumption~\ref{ass:safe_backup_controller} holds and the SafeMPC problem~\eqref{eq:safempc_problem} is feasible at the initial state.
For each state dimension $j=1,\dots,n_x$, assume that the corresponding model error component satisfies the assumptions of Theorem~\ref{thm:uniform_bound}, and that the projection error is bounded as in Proposition~\ref{prop:projection_error_rbf}.
Then the closed-loop system under the DTF-GP SafeMPC controller is probabilistically safe in the sense that
\[
\mathbb P\!\left[
\forall i\in\mathbb N:
x_i\in\mathcal X,\;
u_i\in\mathcal U
\right]
\ge
1-n_x\delta .
\]
\end{cor}

\begin{pf*}{Proof.}
By Theorem~\ref{thm:uniform_bound}, applied component-wise with confidence levels $\delta$, the DTF-GP confidence intervals contain the true model error uniformly over all $z\in\mathcal Z$ and all data sizes $N\ge1$ with probability at least $1-n_x \delta$.
On this event, the uncertainty propagation map $\mathcal M$ over-approximates the true closed-loop successor states.
The constraints in~\eqref{eq:safempc_problem} therefore imply satisfaction of the state and input constraints along the predicted uncertainty sets.
Moreover, the terminal constraint $R_H\subset\mathcal X_{\mathrm{safe}}$ and Assumption~\ref{ass:safe_backup_controller} ensure recursive feasibility through the safe backup controller.
Consequently, the SafeMPC safety argument of~\cite{kollerLearningbasedModelPredictive2018} applies with the DTF-GP confidence bounds in place of the full GP bounds, yielding the stated probability guarantee. \hfill$\blacksquare$
\end{pf*}

\subsection{Safe exploration}
\label{sec:safe_exploration}

We use (DTF-GP) SafeMPC formulation for safe exploration. The objective is chosen to favor trajectories that visit regions of high model uncertainty while satisfying all state, input, and terminal safety constraints. Specifically, the stage cost is selected as
\begin{equation*}
J(R_0,\dots,R_H)
=
-\sum_{t=0}^{H-1}\sum_{j=1}^{n_x}
\sigma_{N,j}(\bar z_t),
\end{equation*}
where $\bar z_t$ denotes the nominal state--input pair associated with the uncertainty set $R_t$, and $\sigma_{N,j}$ is the posterior standard deviation of the GP model for the $j$-th state component. This encourages exploration of uncertain regions while the constraints in~\eqref{eq:safempc_problem} enforce safety.

We consider both static and dynamic exploration. In static exploration, the system is reset after each exploration step and the initial state is optimized jointly with the input sequence.
This setting is used to compare the runtimes of the DTF-GP and the full GP.
In dynamic exploration, the optimized input is applied to the system, and the resulting state is used as the initial condition for the next MPC iteration.
We use this setting to compare exploration performance.
Note that, by construction, the safety guarantee of Corollary~\ref{cor:dtf_safempc_safety} applies to both settings.

\subsection{Gaussian process implementation details}
\label{sec:mpc_implementation}

We implement both a full GP and the DTF-GP using plain \texttt{NumPy} to ensure a fair comparison. Both models use the same data, hyperparameter learning procedure, and differ only in the kernel representation.

\paragraph{Kernel.}
We employ a composite kernel consisting of a linear kernel and an RBF kernel
\[
k(z,z') = k_{\mathrm{lin}}(z,z') + k_{\mathrm{RBF}}(z,z').
\]
This choice has been made to capture also linear functions, which do not belong to the RKHS of the squared-exponential kernel consisting of smooth, stationary functions with rapidly decaying spectra~\cite{steinwartSupportVectorMachines2008}.
For the DTF-GP, the corresponding feature map is obtained by concatenating the linear features with the trigonometric features approximating the RBF component. Since the linear component is represented exactly, the projection error remains unchanged and arises only from the approximation of the RBF component.

\paragraph{Frequency selection.}
\label{sec:impl_frequency_selection}
The RBF kernel is approximated using the trigonometric feature representation introduced in Section~\ref{sec:scalable_gp}, with frequencies selected according to the ellipsoidal truncation~\eqref{eq:frequency_selection}.
The periods $T=(T_1,\dots,T_d)$ are chosen as $1.2$ times the domain length in each dimension to mitigate boundary effects due to the periodicity of the kernel, as discussed in Remark~\ref{rem:periodicity}.
The truncation radius is chosen adaptively based on a prescribed target projection error. In particular, $r$ is computed from the projection error bound in Proposition~\ref{prop:projection_error_rbf} as the smallest value for which the bound falls below the desired target projection error.

\paragraph{Uncertainty bounds.}
For the full GP, we use the high-probability confidence bound introduced in~\eqref{eq:full_gp_bound}. For the DTF-GP, we use the uncertainty bound from Theorem~\ref{thm:uniform_bound}, where the projection error is evaluated using Proposition~\ref{prop:projection_error_rbf}. We conservatively estimate the RKHS norm bound $B$ in a separate offline procedure using noise-free samples of the model mismatch.

\paragraph{Hyperparameter learning.}
For both models, hyperparameters (kernel parameters and noise variance) are learned by maximizing the marginal likelihood on the initial safe dataset and then kept fixed during exploration.
For the DTF-GP, we exploit the structure of the feature representation to cache data-dependent quantities. In particular, the feature matrix can be decomposed into a data-dependent basis evaluation and a parameter-dependent weighting matrix, which enables precomputing and reusing the data-dependent terms during hyperparameter optimization. This reduces the cost of repeated marginal likelihood evaluations from $\mathcal{O}(N M^2 + M^3)$ to $\mathcal{O}(M^3)$.
Further details are provided in Appendix~\ref{app:caching}.

\paragraph{Online updates.}
Online updates are performed during static exploration, where newly collected samples are incorporated after each SafeMPC iteration.
In dynamic exploration, the GP is trained once on the initial safe dataset and kept fixed during the rollout.
For the full GP, adding one sample extends the regularized Gram matrix by one row and column. The Cholesky factor is updated using a triangular solve, resulting in cost $\mathcal O(N^2)$ per sample.
For the DTF-GP, adding one sample corresponds to a rank-one update of the feature-space matrix $V_N$, whose Cholesky factor can be updated with cost $\mathcal O(M^2)$ per sample~\cite[Sec.~6.5.4]{golubMatrixComputations2013}.
Table~\ref{tab:complexities} summarizes the computational complexities of the main GP operations, highlighting the theoretical computational advantage of the DTF-GP when $N \gg M$.

\begin{table}[h!]
\centering
\caption{Computational complexities of the full GP and the DTF-GP, with $N$ denoting the number of data points and $M$ the number of DTF-GP features.}
\label{tab:complexities}
\vspace{0.8em}
\begin{tabular}{lcc}
\hline
Operation & Full GP & DTF-GP \\
\hline
Marginal likelihood evaluation & $\mathcal O(N^3)$ & $\mathcal O(M^3)$ \\
Online update per sample & $\mathcal O(N^2)$ & $\mathcal O(M^2)$ \\
Prediction per test point & $\mathcal O(N)$ & $\mathcal O(M)$ \\
\hline
\end{tabular}
\end{table}

\subsection{Numerical results}
\label{sec:mpc_results}

We evaluate the proposed MPC framework on an inverted pendulum system with dynamics $ml^2\ddot{\vartheta}= m g l \sin(\vartheta) + u$, where $m=0.15\,\mathrm{kg}$ is the pendulum mass, $l=0.5\,\mathrm{m}$ is the pendulum length, and $g=9.81\,\mathrm{m/s^2}$ is the gravitational constant. The state is $x=[\dot{\vartheta},\vartheta]^\top$, and the scalar input is $u$.
The nominal model is obtained by linearizing the system dynamics around the upright equilibrium using a slightly perturbed mass $m_\mathrm{nom}=0.149\,\mathrm{kg}$. This induces a model mismatch that is learned by the GP.
The MPC sampling time is $\Delta t=0.05\,\mathrm{s}$, and the prediction horizon is $H=4$ for dynamic exploration and $H=2$ for static exploration.
The GP models are trained on the domain $\mathcal Z = [-3.0,\,3.0]\times[-1.25,\,1.25]\times[-1.0,\,1.0]$ with observation noise $\sigma_n=[5\cdot10^{-4},\ 5\cdot10^{-5}]^\top$ and RKHS norm bounds $B=[20.0,\,20.0]$.
The projection error targets used for frequency selection (cf.~Section~\ref{sec:impl_frequency_selection}) are set to $[5\!\cdot\!10^{-6},\,5\!\cdot\!10^{-7}]$.
The safe region $\mathcal X_{\mathrm{safe}}$ as assumed in Assumption~\ref{ass:safe_backup_controller} is chosen as the diamond-shaped set
$
\mathcal X_{\mathrm{safe}}
=
\mathrm{conv}\!\left\{
(-1.2,\,20^\circ),\,
(0.8,\,0),\,
(1.2,\,-20^\circ),\,
(-0.8,\,0)
\right\}$.
This safe set, together with the domain bounds, is visualized in Figure~\ref{fig:exploration_trajectories}.
The initial GP models are trained on samples generated inside the safe region. 
The DTF-GP is compared to a full GP in terms of exploration performance and computational efficiency. 
All experiments use confidence level $1-\delta=0.95$. Static exploration results are averaged over $100$ random seeds, while dynamic exploration results are averaged over $50$ random seeds.

\paragraph{Exploration performance.} 
We use mutual information between the collected observations and the unknown model error under a fixed reference GP prior as a measure of sample informativeness~\cite{srinivasGaussianProcessOptimization2010}. Higher mutual information values indicate that the collected data reduce uncertainty about the unknown model error more strongly, and therefore correspond to better exploration performance.
Figure~\ref{fig:exploration_trajectories} shows representative dynamic exploration trajectories for the full GP and the DTF-GP in the state space.
Both methods explore similar regions while respecting the domain bounds and the terminal invariant safe region.
The corresponding quantitative comparison is shown in Figure~\ref{fig:info_gain_dynamic_comparison}, which reports the cumulative mutual information obtained during dynamic exploration. By computing the truncation radius according to a prescribed projection error target (cf. Section~\ref{sec:impl_frequency_selection}), the DTF-GP achieves exploration performance comparable to that of the full GP across all considered dataset sizes. This indicates that the additional conservatism introduced by the finite-dimensional approximation of the kernel can be controlled through a sufficiently large number of retained frequencies.

\begin{figure}
    \centering
    \includegraphics[width=0.5\textwidth]{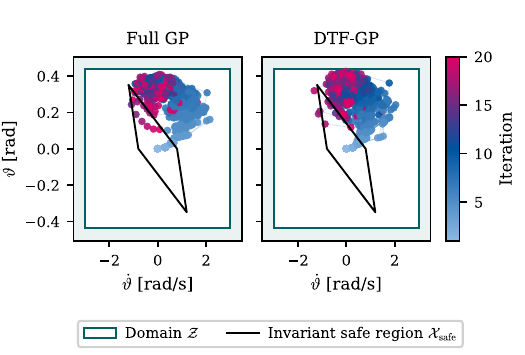}
    \caption{50 representative dynamic exploration trajectories of the full GP and the DTF-GP for $N_{\mathrm{init}}=400$ initial safe samples.}
    \label{fig:exploration_trajectories}
\end{figure}

\begin{figure*}
    \centering
    \includegraphics[width=\textwidth]{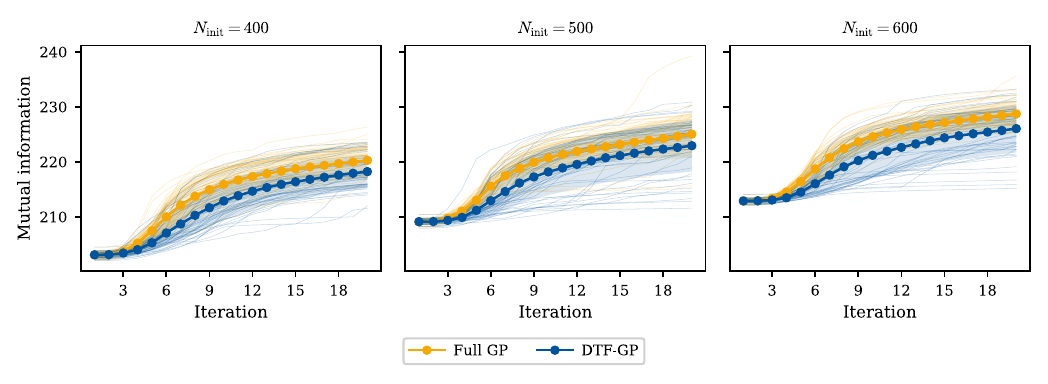}
    \caption{Cumulative mutual information during dynamic exploration for different numbers of initial safe samples $N_{\mathrm{init}}$. Results are averaged over $50$ random seeds.}
    \label{fig:info_gain_dynamic_comparison}
\end{figure*}

\paragraph{Computational advantage.}
The prescribed projection error target leads to a varying number of retained frequencies depending on the estimated kernel hyperparameters (in particular, the spectral decay rates that determine the complexity of the learned function class). Since the computational cost of the DTF-GP depends primarily on the feature dimension, this leads to varying runtimes across different runs. This variability is visible in Figure~\ref{fig:timing_static}, where the DTF-GP exhibits a larger spread of individual runtimes than the full GP. Note that one could instead prescribe a fixed number of retained frequencies, yielding stable runtimes at the expense of varying approximation accuracy and conservatism.
Figure~\ref{fig:timing_static} reports the average computation time per SafeMPC iteration during static exploration, decomposed into GP online updates and MPC solve time, where the latter includes the GP predictions evaluated within the CasADi solver. The DTF-GP shows a larger spread of individual runtimes than the full GP due to the dependence of the feature dimension on the estimated hyperparameters. In contrast to the full GP, the computation time of the DTF-GP exhibits almost no dependence on the number of initial samples. This behavior is consistent with the theoretical complexity of the DTF-GP, which scales with feature dimension rather than dataset size.
Overall, the DTF-GP becomes computationally advantageous for larger datasets while maintaining exploration performance comparable to the full GP.
This provides a practical use case for the DTF-GP in noise-rich or high-confidence-level MPC settings, where large datasets are required to drive uncertainty to low levels.

\begin{figure}
    \centering
    \includegraphics[width=0.45\textwidth]{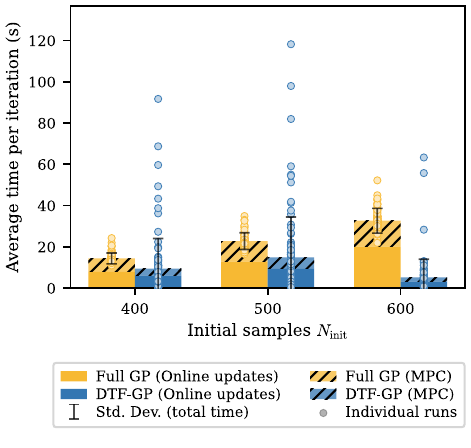}
    \caption{Average computation time per SafeMPC iteration during static exploration as a function of the number of initial safe samples $N_{\mathrm{init}}$ averaged over $100$ random seeds.}
    \label{fig:timing_static}
\end{figure}



\section{Conclusions}
\label{sec:conclusions}

This paper presented the DTF-GP, a scalable GP framework for safe learning-based MPC with high-probability uniform uncertainty guarantees. 
The DTF-GP is based on a finite-dimensional trigonometric kernel approximation using a deterministic frequency grid, which reduces GP inference to Bayesian linear regression in feature space. 
The main theoretical contribution is a uniform uncertainty bound for the DTF-GP. 
For general finite-dimensional kernel approximations, the bound requires an upper bound on the corresponding projection error in addition to the full GP confidence-bound assumptions.
For the RBF kernel, we derived such a projection error bound explicitly, so that the resulting DTF-GP confidence bound relies on the same assumptions as full GP confidence bounds.
We integrated the DTF-GP and its uncertainty bound into a safe learning-based MPC framework, thereby retaining high-probability safety guarantees.
Numerical results on an inverted pendulum system show that the DTF-GP achieves exploration performance comparable to a full GP while reducing computational cost in large-data regimes.
Future work includes extending the analysis to other stationary kernels as well as deriving alternative spectral parametrizations and frequency truncation strategies.

\begin{ack}                               
Anna Scampicchio acknowledges that this work was partially supported by the Wallenberg AI, Autonomous Systems and
Software Program (WASP) funded by the Knut and Alice Wallenberg Foundation. Johanna Menn acknowledges that this work was partially supported by the Deutsche Forschungsgemeinschaft (DFG, German Research Foundation) under Germany's Excellence Strategy -- EXC-2023 Internet of Production -- 390621612, project 556142469 (ParDyBO) and was performed partially within the Helmholtz School for Data Science in Life, Earth and Energy (HDS-LEE). 
\end{ack}

\bibliographystyle{plain}        

\section*{Appendix}

\appendix

The appendix provides supplementary derivations and implementation details omitted from the main text. Appendix~\ref{app:mercer} gives the Mercer interpretation of the DTF kernel. Appendix~\ref{app:periodization} relates the DTF kernel to periodized stationary kernels, and Appendix~\ref{app:periodic_extension} bounds the resulting periodic extension error. Appendices~\ref{app:proof_uniform_bound} and~\ref{app:proof_proj_error} contain the proofs of the uniform uncertainty bound and the RBF projection error bound. Finally, Appendix~\ref{app:caching} describes the caching strategy used for hyperparameter learning.

\section{Mercer decomposition of the DTF kernel}
\label{app:mercer}

Consider the compact domain $\mathcal{Z} \doteq\prod_{j=1}^d [0,T_j]$ equipped with the normalized Lebesgue measure $d\mu(z)=\frac{1}{|\mathcal{Z}|}dz$. The basis functions defined in~\eqref{eq:trig_basis_functions} satisfy
\[
\int_{\mathcal Z}
\varphi_m(z')\varphi_j(z')\,\mathrm d\mu(z') 
=
\delta_{mj}
\]
and therefore form an orthonormal system in the space of square-integrable functions $L^2(\mathcal{Z},\mu)$.

Let $\mathcal T_{k_\mathrm{DTF}}$ denote the integral operator associated with $k_\mathrm{DTF}$:
\[
(\mathcal T_{k_\mathrm{DTF}} f)(z)
=
\int_{\mathcal{Z}}
k_\mathrm{DTF}(z,z')f(z')\,\mathrm d\mu(z').
\]
Using the DTF kernel representation~\eqref{eq:infinite_kernel} and orthonormality of $\{\varphi_m\}_{m\ge0}$, we obtain
\begin{align*}
(\mathcal T_{k_\mathrm{DTF}}\varphi_j)(z)
&=
\sum_{m=0}^{\infty}
\lambda^{(\lceil m/2\rceil)}
\varphi_m(z)
\int_{\mathcal{Z}}
\varphi_m(z')\varphi_j(z')\,\mathrm d\mu(z') \\
&=
\lambda^{(\lceil j/2\rceil)}\varphi_j(z).
\end{align*}
Thus, $\varphi_j$ is an eigenfunction of $\mathcal T_{k_\mathrm{DTF}}$ with eigenvalue $\lambda^{(\lceil j/2\rceil)}$.

\section{Relation of the DTF kernel to periodized kernels}
\label{app:periodization}

In this appendix, we relate the infinite-dimensional DTF kernel introduced in Section~\ref{sec:scalable_gp} to the canonical periodization of a stationary kernel. Let $k$ be a continuous stationary kernel with spectral density $S(\omega)$, i.e.,
\begin{equation*}
k(\tau)
=
\int_{\mathbb R^d}
e^{i2\pi \omega^\top \tau}\,
S(\omega)\,
\mathrm d\omega,
\qquad
\tau=z-z'.
\end{equation*}
For a period vector $T=(T_1,\dots,T_d)$, define the canonical periodization of $k$ by
\begin{equation}
\label{app:periodized_kernel}
k_{\mathrm{per}}(\tau)
\doteq
\sum_{n\in\mathbb Z^d}
k(\tau+n\odot T),
\end{equation}
where $\odot$ denotes the componentwise product. The resulting kernel is periodic with period $T$ in each coordinate and therefore admits a Fourier series representation on the torus
\[
\mathbb T^d
\doteq
\prod_{j=1}^d [0,T_j]
\]
of the form
\begin{equation}
\label{app:fourier_series}
k_{\mathrm{per}}(\tau)
=
\sum_{q\in\mathbb Z^d}
c_q\,
e^{i2\pi \omega_q^\top \tau},
\quad
\omega_q
=
\left(
\frac{q_1}{T_1},
\dots,
\frac{q_d}{T_d}
\right).
\end{equation}
The Fourier coefficients are given by the following expression.
Substituting the definition of $k_{\mathrm{per}}$~\eqref{app:periodized_kernel}, interchanging the sum and the integral, and using a change of variables $s=\tau+n\odot T$ yields
\begin{align*}
c_q
&=
\frac{1}{\prod_{j=1}^d T_j}
\int_{\mathbb T^d}
k_{\mathrm{per}}(\tau)\,
e^{-i2\pi \omega_q^\top \tau}
\,\mathrm d\tau \\
&=
\frac{1}{\prod_{j=1}^d T_j}
\int_{\mathbb T^d}
\sum_{n\in\mathbb Z^d}
k(\tau+n\odot T)\,
e^{-i2\pi \omega_q^\top\tau}
\,\mathrm d\tau \\
&=
\frac{1}{\prod_{j=1}^d T_j}
\sum_{n\in\mathbb Z^d}
\int_{\mathbb T^d}
k(\tau+n\odot T)\,
e^{-i2\pi \omega_q^\top\tau}
\,\mathrm d\tau \\
&=
\frac{1}{\prod_{j=1}^d T_j}
\sum_{n\in\mathbb Z^d}
\int_{\mathbb T^d+n\odot T}
k(s)\,
e^{-i2\pi \omega_q^\top(s-n\odot T)}
\,\mathrm ds \\
&=
\frac{1}{\prod_{j=1}^d T_j}
\sum_{n\in\mathbb Z^d}
\int_{\mathbb T^d+n\odot T}
k(s)\,
e^{-i2\pi \omega_q^\top s} e^{i2\pi q^\top n}
\,\mathrm ds \\
&=
\frac{1}{\prod_{j=1}^d T_j}
\int_{\mathbb R^d}
k(s)\,
e^{-i2\pi \omega_q^\top s}
\,\mathrm ds.
\end{align*}
In the fourth equality, we used the definition of $\omega_q$ in~\eqref{app:fourier_series}, which gives $\omega_q^\top(n\odot T)=q^\top n$, and hence $e^{i2\pi q^\top n}=1$ for $q,n\in\mathbb Z^d$.
By the definition of the spectral density, this gives
\begin{equation*}
c_q
=
\frac{1}{\prod_{j=1}^d T_j}
S(\omega_q).
\end{equation*}
Using Euler's identity and the half-lattice introduced in~\eqref{eq:half_lattice}, the Fourier series~\eqref{app:fourier_series} can be written in real form.
Since the kernel is real-valued and even, the Fourier coefficients satisfy $c_{-q}=c_q\in\mathbb R$.
Hence, recalling that $\tau = z - z'$,
\begin{align}
\label{app:real_fourier_series}
k_{\mathrm{per}}(z,z')
&=
c_0
+
\sum_{q\in\mathbb Z^d_+}
\left(
c_q e^{i2\pi\omega_q^\top\tau}
+
c_{-q} e^{i2\pi\omega_{-q}^\top\tau}
\right) \nonumber\\
&=
c_0
+
\sum_{q\in\mathbb Z^d_+}
c_q
\left(
e^{i2\pi\omega_q^\top\tau}
+
e^{-i2\pi\omega_q^\top\tau}
\right) \nonumber\\
&=
c_0
+
2\sum_{q\in\mathbb Z^d_+}
c_q
\Bigl(
\cos(2\pi\omega_q^\top z)
\cos(2\pi\omega_q^\top z')
\nonumber\\
&\qquad\qquad
+
\sin(2\pi\omega_q^\top z)
\sin(2\pi\omega_q^\top z')
\Bigr).
\end{align}
Comparing~\eqref{app:real_fourier_series} with the kernel representation in~\eqref{eq:infinite_kernel}, we observe that the two kernels $k_{\mathrm{per}}$ and $k_\mathrm{DTF}$ coincide if the weights are chosen as
\begin{equation*}
\label{eq:app_weights_periodic}
\lambda_q
=
c_q,
\qquad
q\in\mathbb Z^d_+.
\end{equation*}
Since the proposed parametrization of the weights (cf. Section~\ref{sec:spectral_parametrization}) mirrors the decay structure of the spectral density, the kernel $k_\mathrm{DTF}$ can in principle reproduce the canonical periodization $k_\mathrm{per}$ of the original stationary kernel. In particular, for kernels with rapidly decaying tails, such as the RBF kernel, the error introduced by the canonical periodization becomes negligible when the period $T$ is chosen sufficiently larger than the domain of interest. We provide a quantitative characterization of this phenomenon in Appendix~\ref{app:periodic_extension} below. 

\section{Periodic extension and approximation of non-periodic functions}
\label{app:periodic_extension}

The DTF kernel introduced in Section~\ref{sec:scalable_gp} is periodic by construction. Consequently, all functions contained in the associated RKHS are periodic as well. In many applications, however, the unknown function of interest is not periodic on the domain where predictions are required. In this appendix, we show that the periodicity assumption induces a controllable approximation error when the period is chosen sufficiently larger than the domain of interest. We present the argument for the one-dimensional squared-exponential kernel, although the same reasoning extends to higher dimensions and to other kernels with rapidly decaying tails.

Let 
$\mathcal Z=[-L/2,L/2]$
denote the domain of interest and let $T>L$ denote the chosen period.
We consider the one-dimensional version of the RBF kernel defined in~\eqref{eq:rbf_kernel} with prior variance $\sigma_f^2$ and lengthscale $\ell_1=\ell>0$.
Assume that the unknown function admits a kernel expansion
\begin{equation*}
\label{eq:app_kernel_expansion}
g(z)
=
\sum_{i=1}^\infty
w_i
k(z_i,z),
\qquad
\sum_{i=1}^\infty |w_i|<\infty,
\end{equation*}
for representing points $z_i\in\mathcal Z$.
We define the canonical periodic extension
\begin{equation*}
\label{eq:app_periodic_extension}
g_{\mathrm{per}}(z)
\doteq
\sum_{n\in\mathbb Z}
g(z+nT).
\end{equation*}
By construction, $g_{\mathrm{per}}$ is $T$-periodic. Moreover, since the canonical periodized kernel $k_{\mathrm{per}}$ coincides with the DTF kernel $k_\mathrm{DTF}$ for appropriately chosen spectral weights (cf.~Appendix~\ref{app:periodization}), the function $g_{\mathrm{per}}$ lies in the RKHS induced by $k_\mathrm{DTF}$ and can therefore be represented by the trigonometric feature expansion introduced in Section~\ref{sec:scalable_gp}.
The difference between $g$ and its periodic extension satisfies
\begin{align*}
|g(z)-g_{\mathrm{per}}(z)|
&=
\left|
\sum_{n\in\mathbb Z\setminus\{0\}}
g(z+nT)
\right| \\
&\le
\sum_{n\in\mathbb Z\setminus\{0\}}
\sum_{i=1}^\infty
|w_i|\,
|k(z_i,z+nT)|.
\end{align*}
For $z\in[-L/2,L/2]$ and $|n| \ge 1$, we obtain
\[
|z_i-(z+nT)|
\ge
|n|T-L.
\]
Using the squared-exponential kernel form therefore yields
\[
|k(z_i,z+nT)|
\le
\sigma_f^2
\exp\!\left(
-\frac{(|n|T-L)^2}{2\ell^2}
\right).
\]
Substituting this estimate gives
\begin{align*}
\label{eq:app_periodic_error_sum}
|g(z)-g_{\mathrm{per}}(z)|
&\le
\sigma_f^2
\sum_{i=1}^\infty |w_i|
\sum_{n\in\mathbb Z\setminus\{0\}}
\exp\!\left(
-\frac{(|n|T-L)^2}{2\ell^2}
\right) \\
&=
2\sigma_f^2
\sum_{i=1}^\infty |w_i|
\sum_{n=1}^\infty
\exp\!\left(
-\frac{(nT-L)^2}{2\ell^2}
\right).
\end{align*}
Since the summand is positive and monotonically decreasing for $n\ge 1$, the tail after the first term can be bounded by
\begin{align*}
&\sum_{n=1}^\infty
\exp\!\left(
-\frac{(nT-L)^2}{2\ell^2}
\right)\\
&\le
\exp\!\left(
-\frac{(T-L)^2}{2\ell^2}
\right) +
\frac{1}{T}
\int_{T-L}^{\infty}
\exp\!\left(
-\frac{u^2}{2\ell^2}
\right)\,\mathrm du.
\end{align*}

Using the Gaussian tail bound, also known as Mills' inequality,
\[
\int_a^\infty e^{-u^2}\,\mathrm du
\le
\frac{1}{2a}
e^{-a^2},
\]
we obtain
\begin{align}
&\frac{1}{T}
\int_{T-L}^{\infty}
\exp\!\left(
-\frac{u^2}{2\ell^2}
\right)\,\mathrm du \notag \\
&\le
\frac{\ell^2}{T(T-L)}
\exp\!\left(
-\frac{(T-L)^2}{2\ell^2}
\right).\notag
\end{align}
Therefore,
\begin{align*}
\label{eq:app_periodic_error_final}
&|g(z)-g_{\mathrm{per}}(z)| \\
&\qquad \quad \le
2\sigma_f^2 
C_w
\left(1+\frac{\ell^2}{T(T-L)}\right)
\exp\!\left(
-\frac{(T-L)^2}{2\ell^2}
\right),
\end{align*}
where
\[
C_w \doteq \sum_{i=1}^\infty |w_i|.
\]
Hence, the approximation error induced by the periodic extension decays exponentially as the distance between the domain of interest and the period increases. Consequently, by choosing the period moderately larger than the domain, the periodic DTF kernel provides an accurate approximation of non-periodic functions generated by the original stationary kernel.

\section{Proof of Theorem~\ref{thm:uniform_bound}}
\label{app:proof_uniform_bound}

As in~\cite{chowdhuryKernelizedMultiarmedBandits2017}, we start from the observation model~\eqref{eq:model_error_observations} and decompose the output vector as $y_{1:N}=g_{1:N}+\varepsilon_{1:N}$. This yields
\begin{align*}
|g(z)-\mu_N(z)|
&=
\left|g(z)-\phi(z)^\top V_N^{-1}\Phi_N^\top y_{1:N}\right| \\
&\le \eta(z)+\mathcal E(z),
\end{align*}
where
\begin{align*}
\eta(z)
&\doteq
\left|\phi(z)^\top V_N^{-1}\Phi_N^\top \varepsilon_{1:N}\right|,\\
\mathcal E(z)
&\doteq
\left|g(z)-\phi(z)^\top V_N^{-1}\Phi_N^\top g_{1:N}\right|.
\end{align*}
The term $\eta(z)$ captures the effect of measurement noise, while $\mathcal{E}(z)$ is a deterministic approximation error arising from finite-dimensional regression.
We first bound the noise term. By the Cauchy--Schwarz inequality in the $V_N^{-1}$-weighted norm,
\begin{align*}
\eta(z)
&=
\left|\phi(z)^\top V_N^{-1}\Phi_N^\top \varepsilon_{1:N}\right| \\
&\le
\|\phi(z)\|_{V_N^{-1}}
\left\|\Phi_N^\top \varepsilon_{1:N}\right\|_{V_N^{-1}},
\end{align*}
where $\|\mathfrak v\|_A=\sqrt{\mathfrak v^\top A \mathfrak v}$ for some vector $\mathfrak v$ and a positive definite matrix $A$ of suitable dimension.
Recalling the posterior variance~\eqref{eq:feature_posterior_variance} of the finite-dimensional GP, we have
\begin{equation}
\label{app:std_dev_bound}
\|\phi(z)\|_{V_N^{-1}} = \frac{\sigma_N(z)}{\sigma_n}.
\end{equation}
It remains to bound $\|\Phi_N^\top \varepsilon_{1:N}\|_{V_N^{-1}}$. In the notation of the self-normalized concentration inequality in~\cite[Corollary~3.5]{abbasiyadkoriOnlineLearningLinearly2012} and~\cite[Theorem~1]{whitehouseSublinearRegretGPUCB2023}, we identify the predictable feature process with $\phi(z_i)$ and set
\[
S_N \doteq \sum_{i=1}^N \phi(z_i)\varepsilon_i
= \Phi_N^\top \varepsilon_{1:N}.
\]
Moreover, using the definition of $V_N$ in~\eqref{eq:regularized_feature_matrix},
\[
V_N = \sigma_n^2 I + \sum_{i=1}^N \phi(z_i)\phi(z_i)^\top
= \sigma_n^2 I + \Phi_N^\top\Phi_N .
\]
Since $(z_i)$ is predictable and $(\varepsilon_i)$ is conditionally $R$-sub-Gaussian given $\mathcal F_{i-1}$, the self-normalized inequality~\cite[Lemma~1]{whitehouseSublinearRegretGPUCB2023} implies that, with probability at least $1-\delta$, simultaneously for all $N\ge 1$,
\begin{equation*}
\left\|\Phi_N^\top \varepsilon_{1:N}\right\|_{V_N^{-1}}
\le
R
\sqrt{
2\ln\!\left(
\frac{1}{\delta}
\sqrt{
\det\!\left(
I+\frac{1}{\sigma_n^2}\Phi_N^\top\Phi_N
\right)}
\right)
}.
\end{equation*}
Consequently,
\begin{equation}
\label{eq:app_noise_bound}
\eta(z)
\le
\sigma_N(z)
\frac{R}{\sigma_n}
\sqrt{
2\ln\!\left(
\frac{1}{\delta}
\sqrt{
\det\!\left(
I+\frac{1}{\sigma_n^2}\Phi_N^\top\Phi_N
\right)}
\right)
}.
\end{equation}

We now bound the approximation term. By adding and subtracting both $Pg$ and $\phi(z)^\top V_N^{-1}\Phi_N^\top Pg_{1:N}$, we obtain
\begin{align}
\label{eq:app_approximation_term}
\mathcal E(z)
&\le |g(z)-Pg(z)| \nonumber\\
&\quad+
\left|Pg(z)-\phi(z)^\top V_N^{-1}\Phi_N^\top Pg_{1:N}\right|\nonumber\\
&\quad+
\left|\phi(z)^\top V_N^{-1}\Phi_N^\top(Pg_{1:N}-g_{1:N})\right|.
\end{align}
We next derive explicit bounds for the second and third terms in~\eqref{eq:app_approximation_term}.
Since $Pg \in \mathcal H_{\tilde{k}_\mathrm{DTF}}$, there exists a coefficient vector
\mbox{$\theta \in \mathbb R^M$} such that 
\[
Pg(z) = \phi(z)^\top \theta=\sum_{m=0}^{M-1}\theta_m \sqrt{\lambda^{(\lceil m/2\rceil)}} \varphi_m (z).
\]
Moreover, the RKHS norm induced by the finite-dimensional kernel is given by
\[
\|Pg\|_{\mathcal H_{\tilde{k}_\mathrm{DTF}}} = \sum_{m=0}^{M-1} \frac{(\theta_m\sqrt{\lambda^{(\lceil m/2\rceil)}})^2}{\lambda^{(\lceil m/2\rceil)}} = \|\theta\|_2.
\]
Hence,
\begin{align}
\label{app:chowdhury_equivalent}
&\left|Pg(z)-\phi(z)^\top V_N^{-1}\Phi_N^\top Pg_{1:N}\right| \nonumber\\
&=
\left|\phi(z)^\top \theta
-
\phi(z)^\top V_N^{-1}\Phi_N^\top\Phi_N \theta
\right| \nonumber\\
&=
\left|\phi(z)^\top
\left(I - V_N^{-1}\Phi_N^\top\Phi_N\right)\theta
\right| \nonumber\\
&\stackrel{\eqref{eq:regularized_feature_matrix}}{=}
\left|\sigma_n^2 \phi(z)^\top V_N^{-1}\theta\right| \nonumber\\
&\le
\sqrt{\sigma_n^2\phi(z)^\top V_N^{-1}\sigma_n^2V_N^{-1}\phi(z)}
\|\theta\|_2 \nonumber\\
&\le
\sqrt{\sigma_n^2\phi(z)^\top V_N^{-1}V_NV_N^{-1}\phi(z)}
\|\theta\|_2 \nonumber\\
&=\sigma_N(z)\|Pg\|_{\mathcal H_{\tilde{k}_\mathrm{DTF}}}.
\end{align}
Since $P$ is the orthogonal projection onto the truncated RKHS $\mathcal H_{\tilde{k}_\mathrm{DTF}}$ as defined by~\eqref{eq:projection}, and recalling the RKHS norm~\eqref{eq:rkhs_norm}, we have
\begin{align}
\label{app:rkhs_bound_vector}
\|Pg\|_{\mathcal H_{\tilde{k}_\mathrm{DTF}}}^2
&=
\sum_{m=0}^{M-1} \frac{\alpha_m^2}{\lambda^{(\lceil m/2\rceil)}} \nonumber\\
&\le
\sum_{m=0}^\infty \frac{\alpha_m^2}{\lambda^{(\lceil m/2\rceil)}} \nonumber\\
&= \|g\|_{\mathcal H_{k_\mathrm{DTF}}}^2 \nonumber\\
&\le B^2.
\end{align}
Substituting \eqref{app:rkhs_bound_vector} into \eqref{app:chowdhury_equivalent} gives
\begin{equation}
\label{eq:app_rkhs_term}
\left|Pg(z)-\phi(z)^\top V_N^{-1}\Phi_N^\top Pg_{1:N}\right|
\le
\sigma_N(z)B.
\end{equation}
Let $r_N \doteq Pg_{1:N}-g_{1:N}$. By Cauchy--Schwarz and \eqref{app:std_dev_bound},
\begin{align}
\label{eq:app_proj_residual_bound}
\bigl|\phi(z)^\top V_N^{-1}\Phi_N^\top r_N\bigr|
&\le
\|\phi(z)\|_{V_N^{-1}}
\bigl\|\Phi_N^\top r_N\bigr\|_{V_N^{-1}} \nonumber\\
&=
\frac{\sigma_N(z)}{\sigma_n}
\bigl\|\Phi_N^\top r_N\bigr\|_{V_N^{-1}} .
\end{align}
Substituting~\eqref{eq:app_rkhs_term} and~\eqref{eq:app_proj_residual_bound} in~\eqref{eq:app_approximation_term} yields
\begin{align}
\label{eq:app_approximation_bound}
\mathcal E(z)
&\le
|g(z)-Pg(z)|
+
\sigma_N(z)B \nonumber\\
&\quad+
\frac{\sigma_N(z)}{\sigma_n}
\bigl\|\Phi_N^\top(Pg_{1:N}-g_{1:N})\bigr\|_{V_N^{-1}} .
\end{align}
Combining~\eqref{eq:app_noise_bound} and~\eqref{eq:app_approximation_bound} gives~\eqref{eq:final_uniform_bound}.

\section{Proof of Proposition~\ref{prop:projection_error_rbf}}
\label{app:proof_proj_error}

From the expansion of $g$ and the definition of $P$,
\begin{align*}
|g(z)-Pg(z)|
&=
\left|
\sum_{m=0}^{\infty} \alpha_m \varphi_m(z)
-
\sum_{m=0}^{M-1} \alpha_m \varphi_m(z)
\right| \\
&=
\left|
\sum_{m=M}^{\infty} \alpha_m \varphi_m(z)
\right| \\
&\le
\sqrt{2} \sum_{m=M}^{\infty} |\alpha_m| \\
&\le
\sqrt{2}
\left(
\sum_{m=M}^{\infty}
\frac{\alpha_m^2}{\lambda^{(\lceil m/2\rceil)}}
\right)^{1/2} \\
&\qquad \times
\left(
\sum_{m=M}^{\infty}
\lambda^{(\lceil m/2\rceil)}
\right)^{1/2} \\
&\le
\sqrt{2}\, B
\left(
\sum_{m=M}^{\infty}
\lambda^{(\lceil m/2\rceil)}
\right)^{1/2}.
\end{align*}

Since each frequency $\omega_q$ corresponds to two basis functions (sine and cosine), it holds that
\[
\sum_{m=M}^{\infty}
\lambda^{(\lceil m/2\rceil)}
=
\sum_{q^\top \widetilde D q > r^2} 2\lambda_q,
\]
and hence
\begin{equation}
\label{eq:app_tail_sum}
\sup_{z\in\mathcal Z}|g(z)-Pg(z)|
\le
2B
\left(
\sum_{q^\top \widetilde D q > r^2} \lambda_q
\right)^{1/2}.
\end{equation}

Using the parametrization of the weights~\eqref{eq:lambda_decay}, it remains to bound the spectral tail:
\begin{equation}
\label{eq:lambda_sum}
\sum_{q^\top \widetilde D q > r^2} \lambda_q
=
C \sum_{q^\top \widetilde D q > r^2}
e^{-q^\top \widetilde D q}.
\end{equation}

Let $u_q \doteq \widetilde D^{1/2} q$. Then $q^\top \widetilde D q = \|u_q\|^2$, and the points $\{u_q\}_{q\in\mathbb Z^d}$ form a lattice in $\mathbb R^d$ with
fundamental cell\[
\mathcal P = \widetilde D^{1/2}\!\left([-\tfrac12,\tfrac12]^d\right),
\quad
\mathrm{Vol}(\mathcal P) = \sqrt{\det(\widetilde D)}.
\]
Let
\[
\rho \doteq \max_{v\in \mathcal P} \|v\|
= \frac{1}{2}\sqrt{\operatorname{tr}(\widetilde D)}.
\]
denote the maximum radius of the cell $\mathcal P$.

Since for all $u \in u_q + \mathcal P$, $\|u_q\|\ge\|u\|-\rho$ it holds that
\[
e^{-\|u_q\|^2}
\le
e^{-(\|u\|-\rho)^2},
\]
then we obtain the bound
\begin{align*}
e^{-\|u_q\|^2}
&\le
\frac{1}{\operatorname{Vol}(\mathcal P)}
\int_{u_q+\mathcal P}
e^{-\|u_q\|^2}\,du \\
&\le
\frac{1}{\operatorname{Vol}(\mathcal P)}
\int_{u_q+\mathcal P}
e^{-(\|u\|-\rho)^2}\,du .
\end{align*}
Summing over all omitted lattice points yields
\begin{align*}
\sum_{\|u_q\|>r} e^{-\|u_q\|^2}
&\le
\frac{1}{\operatorname{Vol}(\mathcal P)}
\sum_{\|u_q\|>r}
\int_{u_q+\mathcal P}
e^{-(\|u\|-\rho)^2}\,du \\
&=
\frac{1}{\operatorname{Vol}(\mathcal P)}
\int_{\bigcup_{\|u_q\|>r}(u_q+\mathcal P)}
e^{-(\|u\|-\rho)^2}\,du .
\end{align*}
The union of the shifted cells satisfies
\[
\bigcup_{\|u_q\|>r}(u_q+\mathcal P)
\subseteq
\{u\in\mathbb R^d:\|u\|>r-\rho\}.
\]
Thus, 
\[
\sum_{\|u_q\|>r} e^{-\|u_q\|^2}
\le
\frac{1}{\operatorname{Vol}(\mathcal P)}
\int_{\|u\|>r-\rho}
e^{-(\|u\|-\rho)^2}\,du .
\]

Substituting into~\eqref{eq:lambda_sum}, we obtain
\[
\sum_{q^\top \widetilde D q > r^2} \lambda_q
\le
\frac{C}{\sqrt{\det(\widetilde D)}}
\int_{\|u\|>r-\rho}
e^{-(\|u\|-\rho)^2}\,du.
\]

Since the integrand depends only on $\|u\|$, we use polar coordinates and obtain
\begin{align*}
\int_{\|u\|>r-\rho}
e^{-(\|u\|-\rho)^2}\,du
&=
S_{d-1}
\int_{r-\rho}^{\infty}
e^{-(t-\rho)^2} t^{d-1}\,dt,
\end{align*}
where $S_{d-1}$ denotes the surface area of the unit sphere in $\mathbb R^d$,
\begin{equation}
\label{eq:surface_area}
S_{d-1}
=
\frac{2\pi^{d/2}}{\Gamma(d/2)},
\end{equation}
and $\Gamma(\cdot)$ is the Gamma function.

Combining the bounds with~\eqref{eq:app_tail_sum} gives
\begin{align*}
&\sup_{z\in\mathcal Z}|g(z)-Pg(z)| \\
&\qquad \quad \le
2B
\Bigg(
\frac{C}{\sqrt{\det(\widetilde D)}}\,
S_{d-1}
\int_{r-\rho}^{\infty}
e^{-(t-\rho)^2}t^{d-1}\,\mathrm dt
\Bigg)^{1/2}.
\end{align*}

\section{Caching for hyperparameter learning}
\label{app:caching}

Recall from~\eqref{eq:phi_trig_features} that the DTF-GP feature map consists of trigonometric basis functions weighted by the square roots of the spectral weights. We write the feature matrix as
\begin{equation}
\label{eq:app_feature_factorization}
    \Phi_N(\xi)
    =
    H_N L(\xi),
\end{equation}
where $H_N\in\mathbb R^{N\times M}$ contains the evaluations of the unweighted trigonometric basis functions at the training inputs and $L(\xi)\in\mathbb R^{M\times M}$ is a diagonal matrix containing the parameter-dependent feature weights, with $\xi$ being the vector of hyperparameters, consisting of the kernel parameters and the noise variance.

The feature-space posterior and marginal likelihood depend on the data through the quantities
\begin{equation*}
    \Phi_N(\xi)^\top \Phi_N(\xi),
    \qquad
    \Phi_N(\xi)^\top y_{1:N}.
\end{equation*}
Using~\eqref{eq:app_feature_factorization}, these quantities can be written as
\begin{align*}
\Phi_N(\xi)^\top \Phi_N(\xi)
&=
L(\xi) H_N^\top H_N L(\xi), \\
\Phi_N(\xi)^\top y_{1:N}
&=
L(\xi) H_N^\top y_{1:N}.
\end{align*}
Thus, the data-dependent terms
\begin{equation*}
\label{eq:app_cache_definitions}
    G_N \doteq H_N^\top H_N,
    \qquad
    v_N \doteq H_N^\top y_{1:N}
\end{equation*}
can be precomputed once and reused during hyperparameter optimization. For a given value of $\xi$, the feature-space matrix becomes
\begin{equation}
\label{eq:app_feature_space_matrix_cached}
    V_N(\xi)
    =
    L(\xi)G_NL(\xi)+\sigma_n^2 I,
\end{equation}
and the data vector entering the posterior mean is
\begin{equation}
\label{eq:app_feature_space_vector_cached}
    \Phi_N(\xi)^\top y_{1:N}
    =
    L(\xi)v_N.
\end{equation}

This caching avoids recomputing $\Phi_N(\xi)^\top\Phi_N(\xi)$ and $\Phi_N(\xi)^\top y_{1:N}$ from all training samples for every marginal likelihood evaluation. Without caching, forming $\Phi_N(\xi)^\top\Phi_N(\xi)$ costs $\mathcal O(NM^2)$, followed by a Cholesky factorization of an $M\times M$ matrix with cost $\mathcal O(M^3)$. With caching, each evaluation only requires assembling~\eqref{eq:app_feature_space_matrix_cached} from the cached matrix $G_N$ and factorizing the resulting $M\times M$ matrix. Hence, the per-evaluation cost is reduced from $\mathcal O(NM^2+M^3)$ to $\mathcal O(M^3)$.

\end{document}